\documentclass[12pt,journal,onecolumn]{IEEEtran}
\usepackage{pslatex}
\usepackage{amsfonts,color,morefloats}
\usepackage{amssymb,amsmath,latexsym,amsthm}

\newcommand{\lcm}{{\rm lcm}}

\newcommand{\Z}{\mathbb{{Z}}}

\newcommand{\gf}{{\mathrm{GF}}}

\newcommand{\C}{{\mathcal{C}}}

\newcommand{\RM}{{\mathrm{RM}}}

\newcommand{\bzero}{{\mathbf{0}}}

\newcommand{\code}{\C_{(q,m,\delta)}}

\newcommand{\tcode}{\tilde{\C}_{(q,m,\delta)}}
\newcommand{\ocode}{\overline{\C}_{(q,m,\delta)}}
\newcommand{\gx}{\overline{g}_{(q,m,\delta)}(x)}
\newcommand{\gt}{\tilde{g}_{(q,m,\delta)}(x)}
\newcommand{\gp}{g_{(q,m,\delta)}^+(x)}
\newcommand{\gm}{g_{(q,m,\delta)}^-(x)}

\newcommand{\Fq}{\gf(q)}
\newcommand{\Fqm}{\gf(q^m)}
\newcommand{\Ftm}{\gf(2^m)}
\newcommand{\Ft}{\gf(2)}

\newcommand{\al}{\alpha}

\newcommand{\de}{\delta}
\newcommand{\lam}{\lambda}

\newcommand{\ep}{\epsilon}
\newcommand{\ti}{\tilde}
\newcommand{\supp}{\text{supp}}
\newcommand{\ol}{\overline}

\newcommand{\lf}{\lfloor}
\newcommand{\rf}{\rfloor}

\newcommand{\pr}{\prime}

\newtheorem{theorem}{Theorem}
\newtheorem{lemma}[theorem]{Lemma}
\newtheorem{proposition}[theorem]{Proposition}
\newtheorem{corollary}[theorem]{Corollary}
\newtheorem{conj}{Conjecture}

\newtheorem{example}{Example}
\newtheorem{remark}{Remark}

\begin{document}

\title{A Family of Reversible BCH Codes}

\author{Shuxing~Li, ~Cunsheng~Ding, and Hao Liu    
\thanks{S. Li is with the Department of Mathematics, Hong Kong University of Science and Technology,
                                                  Clear Water Bay, Kowloon, Hong Kong, China  (E-mail: lsxlsxlsx1987@gmail.com).}
\thanks{C. Ding is with the Department of Computer Science
                                                  and Engineering, Hong Kong University of Science and Technology,
                                                  Clear Water Bay, Kowloon, Hong Kong, China (E-mail: cding@ust.hk).}

\thanks{H. Liu is with the Department of Computer Science
                                                  and Engineering, Hong Kong University of Science and Technology,
                                                  Clear Water Bay, Kowloon, Hong Kong, China (E-mail: hliuar@connect.ust.hk).}

}

\date{\today}
\maketitle

\begin{abstract}
Cyclic codes are an interesting class of linear codes due to their efficient encoding and decoding algorithms as well as their theoretical importance. BCH codes form a subclass of cyclic codes and are very important in both theory and practice as they have good error-correcting capability and are widely used in communication systems, storage devices and consumer electronics. However, the dimension and minimum distance of BCH codes are not known in general. The objective of this paper is to study the dimension and minimum distance of a family of BCH codes over finite fields, i.e., a class of reversible BCH codes.
\end{abstract}


\section{Introduction}\label{sec-intro}

Throughout this paper, let $p$ be a prime and let $q$ be a power of $p$.
An $[n,k,d]$ linear code $\C$ over $\gf(q)$ is a $k$-dimensional subspace of $\gf(q)^n$ with minimum
(Hamming) distance $d$. A linear code $\C$ over $\gf(q)$ is called a \emph{linear code with complementary dual (in short, LCD)}
if $\C \cap \C^\perp =\{\bzero\}$, where $\C^\perp$ denotes the dual of $\C$.

An $[n,k]$ linear code $\C$ is called {\em cyclic} if $(c_0,c_1, \cdots, c_{n-1}) \in \C$ implies $(c_{n-1}, c_0, c_1, \cdots, c_{n-2}) \in \C$. By identifying any vector $(c_0,c_1, \cdots, c_{n-1}) \in \gf(q)^n$ with
$$
c_0+c_1x+c_2x^2+ \cdots + c_{n-1}x^{n-1} \in \gf(q)[x]/(x^n-1),
$$
any code $\C$ of length $n$ over $\gf(q)$ corresponds to a subset of the quotient ring $\gf(q)[x]/(x^n-1)$. A linear code $\C$ is cyclic if and only if the corresponding subset is an ideal of the ring $\gf(q)[x]/(x^n-1)$.

Note that every ideal of $\gf(q)[x]/(x^n-1)$ is principal. Let $\C=\langle g(x) \rangle$ be a cyclic code, where $g(x)$ is monic and has the smallest degree among all the elements of $\C$. Then $g(x)$ is unique and called the {\em generator polynomial}, and $h(x)=(x^n-1)/g(x)$ is referred to as the {\em parity-check} polynomial of $\C$. Let $f(x) \in \Fq[x]$ be a polynomial with degree $l$, then the reciprocal polynomial of $f$ is defined to be $f_R(x)=x^lf(x^{-1})$. $f$ is called self-reciprocal if and only if $f(x)=f_R(x)$. A cyclic code $\C$ with generator polynomial $g(x)$ is called \emph{reversible} if $g(x)$ is self-reciprocal. The reversibility implies if
$$
(c_0,c_1,\ldots,c_{n-1}) \in \C,
$$
then
$$
(c_{n-1},c_{n-2},\ldots,c_0) \in \C.
$$
It is easily seen that reversible cyclic codes are LCD codes.

Let $m \geq 1$ be a positive integer and let $n=q^m-1$. Let $\alpha$ be a generator of $\gf(q^m)^*$. For any $i$ with $1 \leq i \leq q^m-2$, let $m_i(x)$ denote the minimal polynomial of $\alpha^i$ over $\gf(q)$. For any $2 \leq \delta < n$, define
$$
g_{(q,m,\delta)}^+(x)=\lcm(m_{1}(x), m_{2}(x),\cdots, m_{\delta-1}(x)),
$$
and
$$
g_{(q,m,\delta)}^-(x)=\lcm(m_{n-\de+1}(x), m_{n-\de+2}(x),\cdots, m_{n-1}(x)).
$$
where $\lcm$ denotes the least common multiple of these minimal polynomials. Let $\C_{(q, \, m, \, \delta)}$ denote the cyclic code of length $n$ with generator polynomial $g_{(q,m,\delta)}^+(x)$. This code $\C_{(q,\,  m,\,  \delta)}$ is called the \emph{narrow-sense primitive BCH code} with \emph{designed distance} $\delta$. The codes $\C_{(q,\,  m,\,  \delta)}$ have been studied in series of  papers, including \cite{AKS,ACS92,AS94,Berlekamp,Charp90,Charp98,YH96,DDZ15,D,KL72,Mann,YF}. However, the dimension and minimum distance of the codes $\C_{(q,\,  m,\,  \delta)}$ are still open in general.

For any $2 \leq \delta < (n+2)/2$, define
$$
\ti{g}_{(q,\, m,\, \delta)}(x)=\lcm(\gm,\gp).
$$
Let $\tcode$ denote the cyclic code of length $n$ with generator polynomial $\gt$. $\tcode$ is a reversible cyclic code
which was studied by Massey for the data storage applications \cite{Ma}. For the code $\tcode$ with length $n$ and minimum distance $d$, it was shown in \cite{TH} that
$$
\begin{cases}
d=\de & \mbox{if $\de \mid n$},\\
d \ge \de+1 & \mbox{otherwise}.
\end{cases}
$$
Moreover, if we consider the even-like subcode of $\tcode$, namely, the code $\ocode$ with length $n$ and generator polynomial
$$
\gx=(x-1)\gt,
$$
its minimum distance is at least $2\de$ by the BCH bound. Hence, a potential great improvement on the minimum distance is expected by considering the even-like subcode of $\tcode$. This intuition motivates us to study the code $\ocode$, which is a reversible BCH code.


Recently, LCD codes have found an important application in cryptography \cite{CG2015}. Moreover, it was shown that asymptotically good LCD codes exist \cite{EY,Massey1,Sendrier}. These are our major motivations of studying reversible BCH codes, which are a special class of LCD codes. The objective of this paper is to study the fundamental parameters, namely dimensions and minimum distances of the reversible BCH codes $\ocode$.

\section{The dimension of the reversible BCH codes $\ocode$}

Our task in this section is to study the dimension of the codes $\ocode$. It is in general a hard problem to determine the dimension of a BCH code.

\subsection{The dimension of $\ocode$ when $\de$ is relatively small}

In \cite{YH96}, the authors derived the dimension of the narrow-sense primitive BCH code $\code$ when $\de$ is relatively small. In this subsection, we develop similar results for the dimension of the reversible BCH code $\ocode$.

Any positive integer $s$ with $0 \le s \le n$ has a unique $q$-ary expansion as $s=\sum_{i=0}^{m-1}s_iq^i$, where $0 \le s_i \le q-1$. We use $C_{\de}$ to denote the $q$-cyclotomic coset in $\Z_n$ containing $\de$, which is defined by
$$
C_{\de}=\{ \delta q^i \bmod{n}: 0 \leq i \leq m-1\}.
$$
We use $cl(\de)$ to denote the coset leader of $C_{\de}$, which is the smallest integer in $C_{\de}$. The $q$-ary expansion sequence of $s=\sum_{i=0}^{m-1}s_iq^i$ is denoted as $\ol{s}=(s_{m-1},s_{m-2},\ldots,s_0)$. Below, we simply call the $q$-ary expansion sequence of $s$ as the sequence of $s$, whenever this causes no confusion. The weight of $\ol{s}$ is defined to be the number of nonzero entries among $\ol{s}$ and denoted by $wt(\ol{s})$. Define the support of $\ol{s}$ as
$$
\supp(\ol{s})=\{0 \le i \le m-1 \mid s_i \ne 0\}.
$$


We have the following properties regarding $q$-cyclotomic cosets modulo $n=q^m-1$.

\begin{lemma}\label{lem-cyccoset1}
Let $m \ge 2$ and $n=q^m-1$. Then we have the following.
\begin{itemize}
\item[1)] When $m$ is odd, for $1 \le i \le q^{(m+1)/2}$, $|C_i|=|C_{-i}|=m$. For $1 \le i,j \le q^{(m+1)/2}$, $j \in C_i$, or equivalently $-j \in C_{-i}$, if and only if
$$
j \equiv q^li \pmod{n}
$$
for some $0 \le l \le m-1$.

\item[2)] When $m$ is even, $|C_{q^{m/2}+1}|=|C_{-q^{m/2}-1}|=\frac{m}{2}$ and $|C_i|=|C_{-i}|=m$ for $1 \le i \le 2q^{m/2}$, $i \ne q^{m/2}+1$. For $1 \le i,j \le 2q^{m/2}$, $j \in C_i$, or equivalently $-j \in C_{-i}$, if and only if
$$
j \equiv q^li \pmod{n}
$$
for some $0 \le l \le m-1$.
\end{itemize}
\end{lemma}
\begin{proof}
1) The proof is essentially the same as that of \cite[Theorem 3]{YH96} (see also {\rm \cite[Lemmas 8,9]{AKS}}), and
is omitted.

2) The proof is essentially the same as that of \cite[Theorem 3]{YH96}, and is omitted.
\end{proof}

\begin{lemma}\label{lem-cyccoset2}
Let $m \ge 2$ and let $n=q^m-1$. Then the following statements hold.
\begin{itemize}
\item[1)] When $m$ is odd, for $1 \le i,j \le q^{(m+1)/2}$, $-j \in C_i$ if and only if
\begin{align*}
\ol{i}&=(0,\ldots,0,\underset{\frac{m-1}{2}}{q-1},\ldots,\underset{1}{q-1},u), \\
\ol{j}&=(0,\ldots,0,\underset{\frac{m-1}{2}}{q-1-u},\underset{\frac{m-3}{2}}{q-1},\ldots,q-1),
\end{align*}
or
\begin{align*}
\ol{i}&=(0,\ldots,0,\underset{\frac{m-1}{2}}{q-1-u},\underset{\frac{m-3}{2}}{q-1},\ldots,q-1), \\
\ol{j}&=(0,\ldots,0,\underset{\frac{m-1}{2}}{q-1},\ldots,\underset{1}{q-1},u),
\end{align*}
where $0 \le u \le q-1$.

\item[2)] When $m$ is even and $q>2$, for $1 \le i,j \le 2q^{m/2}$, $-j \in C_i$ if and only if
\begin{align*}
\ol{i}&=(0,\ldots,0,\underset{\frac{m}{2}}{1},\underset{\frac{m}{2}-1}{q-1},\ldots,q-1,q-2), \\
\ol{j}&=(0,\ldots,0,\underset{\frac{m}{2}}{1},\underset{\frac{m}{2}-1}{q-1},\ldots,q-1,q-2),
\end{align*}
or
\begin{align*}
\ol{i}&=(0,\ldots,0,\underset{\frac{m}{2}-1}{q-1},\ldots,\underset{1}{q-1},q-2), \\
\ol{j}&=(0,\ldots,0,\underset{\frac{m}{2}}{1},\underset{\frac{m}{2}-1}{q-1},\ldots,q-1),
\end{align*}
or
\begin{align*}
\ol{i}&=(0,\ldots,0,\underset{\frac{m}{2}}{1},\underset{\frac{m}{2}-1}{q-1},\ldots,q-1), \\
\ol{j}&=(0,\ldots,0,\underset{\frac{m}{2}-1}{q-1},\ldots,\underset{1}{q-1},q-2),
\end{align*}
or
\begin{align*}
\ol{i}&=(0,\ldots,0,\underset{\frac{m}{2}-1}{q-1},\ldots,q-1), \\
\ol{j}&=(0,\ldots,0,\underset{\frac{m}{2}-1}{q-1},\ldots,q-1).
\end{align*}
\item[3)] When $m$ is even and $q=2$, for $1 \le i,j \le 2^{(m/2)+1}$, $-j \in C_i$ if and only if
\begin{align*}
\ol{i}&=(0,\ldots,0,\underset{\frac{m}{2}-2}{1},\ldots,1), \\
\ol{j}&=(0,\ldots,0,\underset{\frac{m}{2}}{1},\ldots,1),
\end{align*}
or
\begin{align*}
\ol{i}&=(0,\ldots,0,\underset{\frac{m}{2}}{1},\ldots,1), \\
\ol{j}&=(0,\ldots,0,\underset{\frac{m}{2}-2}{1},\ldots,1),
\end{align*}
or
\begin{align*}
\ol{i}&=(0,\ldots,0,\underset{\frac{m}{2}-1}{1},\ldots,1), \\
\ol{j}&=(0,\ldots,0,\underset{\frac{m}{2}-1}{1},\ldots,1),
\end{align*}
or
\begin{align*}
\ol{i}&=(0,\ldots,0,\underset{\frac{m}{2}}{1},0,1,\ldots,1), \\
\ol{j}&=(0,\ldots,0,\underset{\frac{m}{2}}{1},\ldots,1,0,1),
\end{align*}
or
\begin{align*}
\ol{i}&=(0,\ldots,0,\underset{\frac{m}{2}}{1},\ldots,1,0,1), \\
\ol{j}&=(0,\ldots,0,\underset{\frac{m}{2}}{1},0,1,\ldots,1).
\end{align*}
\end{itemize}
\end{lemma}
\begin{proof}
1) If $-j \in C_i$, then there exists an $l$ with $0 \le l \le m-1$, such that $q^li+j\equiv 0\pmod{n}$. Hence,
$$
\ol{q^li+j}=(q-1,q-1,\ldots,q-1).
$$
Since $m=wt(\ol{q^li+j})\le wt(\ol{i})+wt(\ol{j}) \le m+1$, we have
$\{wt(\ol{i}),wt(\ol{j})\}=\{\frac{m-1}{2},\frac{m+1}{2}\}$ or $\{wt(\ol{i}),wt(\ol{j})\}=\{\frac{m+1}{2}\}.$
If $\{wt(\ol{i}),wt(\ol{j})\}=\{\frac{m-1}{2},\frac{m+1}{2}\}$, then clearly, $\supp(\ol{q^li}) \cap \supp(\ol{j})=\emptyset$. Otherwise, if $\supp(\ol{q^li}) \cap \supp(\ol{j}) \ne \emptyset$, there is at least one entry in $\ol{q^li+j}$, which is not $q-1$. Hence, $\ol{q^li}$ and $\ol{j}$ must have the following two forms
\begin{align*}
\ol{q^li}&=(\underset{m-1}{q-1},\ldots,\underset{\frac{m+1}{2}}{q-1},0,\ldots,0),\\
\ol{j}&=(0,\ldots,0,\underset{\frac{m-1}{2}}{q-1},\ldots,q-1),
\end{align*}
or
\begin{align*}
\ol{q^li}&=(\underset{m-1}{q-1},\ldots,\underset{\frac{m-1}{2}}{q-1},0,\ldots,0),\\
\ol{j}&=(0,\ldots,0,\underset{\frac{m-3}{2}}{q-1},\ldots,q-1),
\end{align*}
If $\{wt(\ol{i}),wt(\ol{j})\}=\{\frac{m+1}{2}\}$, then clearly, $|\supp(\ol{q^li}) \cap \supp(\ol{j})|\ge 1$. If $|\supp(\ol{q^li}) \cap \supp(\ol{j})|>1$, then there is at least one entry in $\ol{q^li+j}$, which is not $q-1$. Hence, $|\supp(\ol{q^li}) \cap \supp(\ol{j})|=1$. Therefore, $\ol{q^li}$ and $\ol{j}$ must have the following $2q-4$ forms
\begin{align*}
\ol{q^li}&=(\underset{m-1}{q-1},\ldots,\underset{\frac{m+1}{2}}{q-1},u,0,\ldots,0), \\
\ol{j}&=(0,\ldots,0,\underset{\frac{m-1}{2}}{q-1-u},\underset{\frac{m-3}{2}}{q-1},\ldots,q-1),
\end{align*}
or
\begin{align*}
\ol{q^li}&=(\underset{m-1}{q-1},\ldots,\underset{\frac{m+1}{2}}{q-1},\underset{\frac{m-1}{2}}{0},\ldots,0,u), \\
\ol{j}&=(0,\ldots,0,\underset{\frac{m-1}{2}}{q-1},\ldots,q-1,q-1-u),
\end{align*}
where $1 \le u \le q-2$. Therefore, the conclusion follows.

2) If $-j \in C_i$, then there exists an $l$ with $0 \le l \le m-1$, such that $q^li+j\equiv 0\pmod{n}$. Hence,
$$
\ol{q^li+j}=(q-1,q-1,\ldots,q-1).
$$
Since $m=wt(\ol{q^li+j})\le wt(\ol{i})+wt(\ol{j}) \le m+2$, we must have
$$
\{wt(\ol{i}),wt(\ol{j})\}=\begin{cases}
\{\frac{m}{2}+1\}, \mbox{ or} \\
\{\frac{m}{2},\frac{m}{2}+1\}, \mbox{ or} \\
\{\frac{m}{2}-1,\frac{m}{2}+1\}, \mbox{ or} \\
\{\frac{m}{2}\}.
\end{cases}
$$
If $\{wt(\ol{i}),wt(\ol{j})\}=\{\frac{m}{2}+1\}$, then $|\supp(\ol{q^li}) \cap \supp(\ol{j})|=2$. Hence, $\ol{q^li}$ and $\ol{j}$ must have the following form
\begin{align*}
\ol{q^li}&=(\underset{m-1}{q-1},\ldots,\underset{\frac{m}{2}+1}{q-1},\underset{\frac{m}{2}}{q-2},\underset{\frac{m}{2}-1}{0},\ldots,0,1), \\
\ol{j}&=(0,\ldots,0,\underset{\frac{m}{2}}{1},\underset{\frac{m}{2}-1}{q-1},\ldots,q-1,q-2).
\end{align*}
If $\{wt(\ol{i}),wt(\ol{j})\}=\{\frac{m}{2}, \frac{m}{2}+1\}$, then $|\supp(\ol{q^li}) \cap \supp(\ol{j})|=1$. Hence, $\ol{q^li}$ and $\ol{j}$ must have the following two forms
\begin{align*}
\ol{q^li}&=(\underset{m-1}{q-1},\ldots,\underset{\frac{m}{2}+1}{q-1},\underset{\frac{m}{2}}{q-2},0,\ldots,0), \\
\ol{j}&=(0,\ldots,0,\underset{\frac{m}{2}}{1},\underset{\frac{m}{2}-1}{q-1},\ldots,q-1),
\end{align*}
or
\begin{align*}
\ol{q^li}&=(\underset{m-1}{q-1},\ldots,\underset{\frac{m}{2}}{q-1},\underset{\frac{m}{2}-1}{0},\ldots,0,1), \\
\ol{j}&=(0,\ldots,0,\underset{\frac{m}{2}-1}{q-1},\ldots,\underset{1}{q-1},q-2).
\end{align*}
If $\{wt(\ol{i}),wt(\ol{j})\}=\{\frac{m}{2}-1, \frac{m}{2}+1\}$, then $|\supp(\ol{q^li}) \cap \supp(\ol{j})|=0$. Hence, there is at least one entry in $\ol{(q^li+j)}$ which is not equal to $q-1$.
If $\{wt(\ol{i}),wt(\ol{j})\}=\{\frac{m}{2}\}$, then $|\supp(\ol{q^li}) \cap \supp(\ol{j})|=0$. Hence, $\ol{q^li}$ and $\ol{j}$ must have the following form
\begin{align*}
\ol{i}&=(\underset{m-1}{q-1},\ldots,\underset{\frac{m}{2}}{q-1},0,\ldots,0), \\
\ol{j}&=(0,\ldots,0,\underset{\frac{m}{2}-1}{q-1},\ldots,q-1).
\end{align*}
Therefore, the conclusion follows.

3) The proof is similar to that of 2) and is omitted here.
\end{proof}

As a consequence, we have the following proposition.

\begin{proposition}\label{prop-cyccoset}
Let $m \ge 2$ and $n=q^m-1$.
\begin{itemize}
\item[1)] Suppose $m$ is odd. Then
\begin{align*}
&|\{(cl(i),cl(j)) \mid -j \in C_i, 1\le i,j\le l\}| \\
&= \begin{cases}
0 & \mbox{if $1 \le l \le q^{(m+1)/2}-q$}, \\
2h & \mbox{if $l=q^{(m+1)/2}-q+h$, $1 \le h \le q-2$}, \\
2(q-1) & \mbox{if $q^{(m+1)/2}-1\le l\le q^{(m+1)/2}$}.
\end{cases}
\end{align*}
\item[2)] Suppose $m$ is even and $q>2$. Then
\begin{align*}
&|\{(cl(i),cl(j)) \mid -j \in C_i, 1\le i,j\le l\}| \\
&= \begin{cases}
0 & \mbox{if $1 \le l \le q^{m/2}-2$}, \\
1 & \mbox{if  $q^{m/2}-1\le l\le 2q^{m/2}-3$}, \\
2 & \mbox{if $l=2q^{m/2}-2$}, \\
4 & \mbox{if $2q^{m/2}-1 \le l \le 2q^{m/2}$}.
\end{cases}
\end{align*}
\item[3)] Suppose $m$ is even and $q=2$. Then
\begin{align*}
&|\{(cl(i),cl(j)) \mid -j \in C_i, 1\le i,j\le l\}| \\
&= \begin{cases}
0 & \mbox{if $1 \le l \le 2^{m/2}-2$}, \\
1 & \mbox{if $2^{m/2}-1 \le l \le 2^{(m/2)+1}-4$}, \\
3 & \mbox{if $2^{(m/2)+1}-3 \le l \le 2^{(m/2)+1}-2$}, \\
5 & \mbox{if $2^{(m/2)+1}-1 \le l \le 2^{(m/2)+1}$}.
\end{cases}
\end{align*}
\end{itemize}
\end{proposition}

Combining Lemma~\ref{lem-cyccoset1} and Proposition~\ref{prop-cyccoset}, we have the following theorem.

\begin{theorem}\label{thm-gene}
Let $m \ge 2$ and $n=q^m-1$. Given an integer $2\le\de\le(n+2)/2$, let $\delta_q$ and $\delta_0$ be the unique integers such that $\delta-1 =\delta_q q+\delta_0$, where $0\le\delta_0<q$. Then the dimension $k$ of the code $\ocode$ is given
below.
\begin{itemize}
\item[1)] When $m$ is odd,
$$
k=\begin{cases}
q^m-2-2m(\de_q(q-1)+\de_0) & \mbox{if $\de\le q^{(m+1)/2}-q$,} \\
q^m-2-2m(q^{(m-1)/2}-1)(q-1) & \mbox{if $q^{(m+1)/2}-q+1\le\de\le q^{(m+1)/2}+1$.}
\end{cases}
$$
\item[2)] When $m$ is even and $q>2$,
$$
k=\begin{cases}
q^m-2-2m(\de_q(q-1)+\de_0) & \mbox{if $\de\le q^{m/2}-1$,} \\
q^m-2-2m(\de_q(q-1)+\de_0-\frac{1}{2}) & \mbox{if $q^{m/2}\le\de\le q^{m/2}+1$,} \\
q^m-2-2m(\de_q(q-1)+\de_0-1) & \mbox{if $q^{m/2}+2\le\de\le 2q^{m/2}-2$,} \\
q^m-2-2m(\de_q(q-1)+\de_0-\frac{3}{2}) & \mbox{if $\de=2q^{m/2}-1$,} \\
q^m-2-2m(\de_q(q-1)+\de_0-\frac{5}{2}) & \mbox{if $2q^{m/2}\le\de\le 2q^{m/2}+1$.}
\end{cases}
$$
\item[3)] When $m$ is even and $q=2$,
$$
k=\begin{cases}
2^m-2-2m(\de_q+\de_0) & \mbox{if $\de\le 2^{m/2}-1$ and $m \ge 4$,} \\
2^m-2-2m(\de_q+\de_0-\frac{1}{2}) & \mbox{if $2^{m/2}\le\de\le 2^{m/2}+1$ and $m \ge 4$,} \\
2^m-2-2m(\de_q+\de_0-1) & \mbox{if $2^{m/2}+2\le\de\le2^{(m/2)+1}-3$ and $m \ge 4$,} \\
2^m-2-2m(\de_q+\de_0-2) & \mbox{if $2^{(m/2)+1}-2\le\de\le2^{(m/2)+1}-1$ and $m \ge 6$,} \\
2^m-2-2m(\de_q+\de_0-3) & \mbox{if $2^{(m/2)+1}\le\de\le 2^{(m/2)+1}+1$ and $m \ge 6$.}
\end{cases}
$$
\end{itemize}
In addition, the minimum distance $d$ of the code satisfies that $d\ge2\de$.
\end{theorem}
\begin{proof}
Let $\overline{g}_{(q,m,\delta)}(x)$ be the generator polynomial of $\ocode$. For the dimension of the code, we only prove 2) since the proofs of 1) and 3) are similar. By 2) of Lemma~\ref{lem-cyccoset1}, the degree of $\ol{g}_{(q,m,\delta)}(x)$ equals
$$
1+2m(\de_q(q-1)+\de_0)-\ep m-|\{(cl(i),cl(j))\mid -j\in C_i, 1 \le i,j\le \de-1\}|m,
$$
where
$$
\ep=\begin{cases}
0 & \mbox{if $\de \le q^{m/2}+1$}, \\
1 & \mbox{if $\de \ge q^{m/2}+2$}.
\end{cases}
$$
Together with 2) of Proposition~\ref{prop-cyccoset}, we have
$$
\deg{(\ol{g}_{(q,m,\delta)}(x))}=\begin{cases}
1+2m(\de_q(q-1)+\de_0) & \mbox{if $\de\le q^{m/2}-1$,} \\
1+2m(\de_q(q-1)+\de_0-\frac{1}{2}) & \mbox{if $q^{m/2}\le\de\le q^{m/2}+1$,} \\
1+2m(\de_q(q-1)+\de_0-1) & \mbox{if $q^{m/2}+2\le\de\le 2q^{m/2}-2$,} \\
1+2m(\de_q(q-1)+\de_0-\frac{3}{2}) & \mbox{if $\de=2q^{m/2}-1$,} \\
1+2m(\de_q(q-1)+\de_0-\frac{5}{2}) & \mbox{if $2q^{m/2}\le\de\le 2q^{m/2}+1$.}
\end{cases}
$$
Therefore, the conclusion on the dimension in 2) follows. Moreover, by the BCH bound, $\ocode$ has minimum distance $d \ge 2\de$.
\end{proof}

\begin{remark}
For the code $\ocode$, if
\begin{equation}\label{condi}
\sum_{i=0}^{\de} \binom{n}{i}(q-1)^i > q^{n-k},
\end{equation}
then $d \le 2\de$ by the sphere packing bound. Therefore, the knowledge on the dimension of the code $\ocode$ may provide more precise information on the minimum distance in some cases. As an illustration, we use Theorem~\ref{thm-gene} and the inequality (\ref{condi}) to get some binary codes $\ol{\C}_{(2,m,\de)}$ with $d=2\de$, which are listed in Table~\ref{tab-mindis}.
\end{remark}

\begin{table}[h]
\begin{center}
\caption{Some binary code $\ol{\C}_{(2,m,\de)}$ with $d=2\de$}\label{tab-mindis}
\begin{tabular}{|c|c|} \hline
$m$  &   $\de$  \\ \hline
\{5,6,7\} & \{3\} \\
\{8,9,10,11,12,13\} & \{3,5\} \\
\{14,15,17,17,18,19\} & \{3,5,7\} \\
\{20\} & \{3,5,7,9\} \\ \hline
\end{tabular}
\end{center}
\end{table}

\subsection{The dimension of $\ocode$ when $\de=q^\lam$ and $\frac{m}{2}\le\lam\le m-1$}

In \cite{Mann}, the dimension of the narrow-sense primitive BCH code $\code$ with $\de=q^\lam$ is considered. The author derived two closed formulas concerning the dimension of such code. In this subsection, we use the idea in \cite{Mann} to give an estimate of the dimension of the reversible BCH code $\ocode$ with $\de=q^\lam$, where $\frac{m}{2}\le\lam\le m-1$.

Let $s$ and $r$ be two positive integers. Given a sequence of length $s$ and a fixed integer $a$ with $0 \le a \le q-1$, we say that the sequence contains a straight run of length $r$ with respect to $a$, if it has $r$ consecutive entries formed by $a$. If we view the sequence as a circle where the first and last entry are glued together, we say that the sequence contains a circular run of length $r$ with respect to $a$, if this circle has $r$ consecutive entries formed by $a$. When the specific choice of the integer $a$ does not matter, we simply say that the sequence has a straight or circular run of length $r$. Clearly, a straight run is also a circular run and the converse is not necessarily true. We use $l_r(s)$ to denote the number of sequences of length $s$, which contains a straight run of length $r$. Particularly, we define $l_r(0)=0$. The following is a recursive formula of $l_r(s)$ which was presented in \cite{Mann}.

\begin{proposition}{\rm \cite[p. 155]{Mann}}\label{prop-l}
Let $s$ and $r$ be two nonnegative integers. Then
$$
l_r(s)=\begin{cases}
  \begin{matrix}
  0 & \mbox{if $0\le s<r$}, \\
  1 & \mbox{if $s=r$},\\
  ql_r(s-1)+(q-1)(q^{s-r-1}-l_r(s-r-1)) & \mbox{if $s>r$}.
  \end{matrix}
\end{cases}
$$
\end{proposition}

Throughout the rest of this section, we always assume that $\de=q^\lam$ and $\frac{m}{2}\le\lam\le m-1$. Recall that the narrow-sense BCH code $\code$ has generator polynomial $g_{(q,m,\delta)}^+(x)$. Set $r=m-\lam$. Note that $\de-1$ corresponds to following sequence
$$
\ol{\de-1}=(\underbrace{0,\ldots,0}_{r},\underset{\lam-1}{q-1},q-1,\ldots,q-1).
$$
The key observation in \cite{Mann} is that for $1 \le i \le n-1$, $\al^i$ is a root of $g_{(q,m,\delta)}^+(x)$ if and only if the sequence of $i$ has a circular run of length at least $r$ with respect to $0$. Similarly, note that $n-\de+1$ corresponds to the following sequence
$$
\ol{n-\de+1}=(\underbrace{q-1,\ldots,q-1}_{r},\underset{\lam-1}{0},0,\ldots,0).
$$
Therefore, for $1 \le i \le n-1$, $\al^i$ is a root of $g_{(q,m,\delta)}^-(x)$ if and only if the sequence of $i$ has a circular run of length at least $r$ with respect to $q-1$. The following proposition presents the degree of $g_{(q,m,\delta)}^+(x)$ and $g_{(q,m,\delta)}^-(x)$.

\begin{proposition}{\rm \cite[p. 155]{Mann}}\label{prop-deg}
Set $r=m-\lam$. Then
$$
\deg(g_{(q,m,\delta)}^+(x))=\deg(g_{(q,m,\delta)}^-(x))=l_r(m)-1+(q-1)^2\sum_{u=0}^{r-2}(r-u-1)(q^{m-r-u-2}-l_r(m-r-u-2)).
$$
\end{proposition}

We have the following estimation on the dimension of $\ocode$.

\begin{theorem}
Set $r=m-\lam$. The dimension $k$ of $\ocode$ satisfies
$$
k \ge q^m-2l_r(m)+2l_r(m-r)-2(q-1)^2\sum_{u=0}^{r-2}(r-u-1)(q^{m-r-u-2}-l_r(m-r-u-2)),
$$
and
$$
k \le q^m-2l_r(m)+ml_r(m-r)-2(q-1)^2\sum_{u=0}^{r-2}(r-u-1)(q^{m-r-u-2}-l_r(m-r-u-2)).
$$
\end{theorem}
\begin{proof}
Since $\lam\ge\frac{m}{2}$, we have $m\ge 2r$. Define a set $N=\{1 \le i \le n-1 \mid g_{(q,m,\delta)}^+(\al^i)=g_{(q,m,\delta)}^-(\al^i)=0\}$. Since $\gx=\lcm(g_{(q,m,\delta)}^-(x),x-1,g_{(q,m,\delta)}^+(x))$, we have
$$
\deg(\gx)=\deg(g_{(q,m,\delta)}^+(x))+\deg(g_{(q,m,\delta)}^-(x))+1-|N|.
$$
Since $\deg(g_{(q,m,\delta)}^+(x))$ and $\deg(g_{(q,m,\delta)}^-(x))$ are known by Proposition~\ref{prop-deg}, it suffices to estimate the size of $N$. $N$ contains the number $1 \le i \le n-1$, such that $\ol{i}$ contains two runs of length $r$ with respect to $0$ and $q-1$, where at most one of them is a circular run. Let $N^{\pr}$ be the set of integers $1 \le i \le n-2$ such that the first $r$ entries of $\ol{i}$ is a straight run of length $r$ with respect to $0$ and the last $m-r$ entries contain a straight run of length $r$ with respect to $q-1$. Clearly, we have $|N^{\pr}|=l_r(m-r)$. Note that each element of $N$ is a proper cyclic shift of an element of $N^{\pr}$. Moreover, for each $i\in N^{\pr}$, we have
\begin{equation}\label{ineqn-1}
2\le|\{q^ji \bmod{n} \mid 0 \le j \le m-1\}|\le m,
\end{equation}
which implies
$$
2|N^{\pr}| \le |N| \le m|N^{\pr}|.
$$
Thus, the conclusion follows from a direct computation.
\end{proof}

\section{The minimum distance of the reversible BCH codes $\ocode$}

While it is difficult to determine the dimension of reversible BCH codes in general, it is more difficult to find out the minimum distance of reversible BCH codes \cite{Charp98}. For the code  $\ocode$, the BCH bound $d \geq 2 \delta$ is usually very tight. But it would be better if we can determine the minimum distance exactly. In this section, we determine the minimum distance $d$ of the code $\ocode$ in some special cases.

Given a codeword $c=(c_0,c_1,\ldots,c_{n-1}) \in \code$, we say $c$ is reversible if $(c_{n-1},c_{n-2},\ldots,c_0) \in \code$. Namely, $c \in \code$ is reversible if and only if $c \in \tcode$. The following theorem says that the reversible codeword in $\code$ provides some information on the minimum distance on $\ocode$.

\begin{theorem}\label{thm-mindis}
Let $c(x) \in \code$ be a reversible codeword of weight $w$. If $c(1) \ne 0$, then $\ocode$ contains a codeword $(x-1)c(x)$ whose weight is at most $2w$. Therefore the minimum distance $d$ of $\ocode$ satisfies $d\le 2w$. In particular, if the weight of $c(x)$ is $\de$, then the minimum distance $d=2\de$.
\end{theorem}
\begin{proof}
Since $c(x)$ is reversible and $c(1) \ne 0$, we have $(x-1)c(x) \in \ocode$. The weight of $(x-1)c(x)$ is at most $2w$, which implies $d \le 2w$. In particular, if $w=\de$, together with the BCH bound, we have $d=2\de$.
\end{proof}



Let $c(x)=\sum_{i=0}^{n-1}c_ix^i$ be a codeword of a cyclic code $\C$ with length $n=q^m-1$. We can use the elements of $\Fqm^*$ to index the coefficients of $c(x)$. Similarly, let $\ol{\C}$ be the extended cyclic code of $\C$ and let $c(x)=\sum_{i=0}^{n}c_ix^i$ be a codeword of $\ol{\C}$ with length $q^m$. We can use the elements of $\Fqm$ to index the coefficients of $\ol{c(x)}$. The support of $c(x)$ (resp. $\ol{c(x)}$) is defined to be the set of elements in $\Fqm^*$ (resp. $\Fqm$), which correspond to the nonzero coefficients of $c(x)$ (resp. $\ol{c(x)}$).

Given a prime power $q$ and an integer $0 \le s \le q^m-1$, $s$ has a unique $q$-ary expansion $s=\sum_{i=0}^{m-1} s_iq^i$. The $q$-weight of $s$ is defined to be $wt_q(s)=\sum_{i=0}^{m-1}s_i$. Suppose $H$ is a subset of $\Fq^*$, then we use $H^{(-1)}$ to denote the subset $\{h^{-1} \mid h \in H\}$.

The following are two classes of reversible BCH codes whose minimum distances are known.

\begin{corollary}\label{cor-distance}
For the reversible BCH code $\ocode$, we have $d=2 \delta$ if $\de \mid n$.
\end{corollary}
\begin{proof}
It suffices to find a codeword $c(x)$ satisfying the condition in Theorem~\ref{thm-mindis}. If $\de \mid n$, by the proof of \cite[Theorem]{TH}, $\code$ contains a reversible codeword $c(x)$ with weight $\de$, where $c(x)=\sum_{i=0}^{\de-1}a_ix^{\frac{ni}{\de}}$ and $c(1)\ne 0$. The desired conclusion then follows from Theorem \ref{thm-mindis}.
\end{proof}

\begin{corollary}
Let $\de=2^r-1$ and $V$ be an $m$-dimensional vector space over $\Ft$. Suppose $1 \le r \le \lf \frac{m}{2} \rf$, then we can choose four $r$-dimensional subspaces of $V$, say $H_i$, $1 \le i \le 4$ of $V$, such that
$H_1 \cap H_2=\{0\}$ and $H_3 \cap H_4=\{0\}$. If $((H_1 \cup H_2) \setminus \{0\})^{(-1)}=(H_3 \cup H_4) \setminus \{0\}$, then the reversible BCH code $\overline{\C}_{(2,m,\delta)}$ has minimum distance $d=2\de$.
\end{corollary}

\begin{proof}
Define $\C^+_{(2,m,\delta)}$ (resp. $\C^-_{(2,m,\delta)}$) to be the BCH code with length $n$ and generator polynomial $g_{(2,m,\delta)}^+(x)$ (resp. $g_{(2,m,\delta)}^-(x)$). Let $\al$ be a primitive element of $\Ftm$. We can assume the zeros of $\C^+_{(2,m,\delta)}$ (resp. $\C^-_{(2,m,\delta)}$) include the elements $\{\al^i \mid 1 \le i \le \de-1\}$ (resp. $\{\al^{-i} \mid 1 \le i \le \de-1\}$).

The BCH code $\C^+_{(2,m,\delta)}$ (resp. $\C^-_{(2,m,\delta)}$) contains a punctured Reed-Muller code $\RM^+(m-r,m)^*$ (resp. $\RM^-(m-r,m)^*$) as a subcode, in which $\RM^+(m-r,m)^*$ has zeros
$$
\{\al^i \mid 0 < i < 2^m-1, wt_2(i)<r\}
$$
and $\RM^-(m-r,m)^*$ has zeros
$$
\{\al^{-i} \mid 0 < i < 2^m-1, wt_2(i)<r\}.
$$
Let $c=(c_0,c_1,\ldots,c_{n-1})$ be a codeword of $\RM^+(m-r,m)^*$. Since $\RM^+(m-r,m)^*$ is a cyclic code, its coordinates can be indexed in the following way
\begin{equation}\label{eqn-1}
c=(\underset{1}{c_0},\underset{\al}{c_1},\ldots,\underset{\al^{n-1}}{c_{n-1}}),
\end{equation}
where $\sum_{j=0}^{n-1} c_j\al^{ij}=0$ for each $1 \le i \le \de-1$. Similarly, suppose $c^{\pr}=(c_0^{\pr},c_1^{\pr},\ldots,c_{n-1}^{\pr})$ is a codeword of $\RM^-(m-r,m)^*$. Then, its coordinates can be indexed in the following way
\begin{equation}\label{eqn-2}
c^{\pr}=(\underset{1}{c_0^{\pr}},\underset{\al^{-1}}{c_1^{\pr}},\ldots,\underset{\al^{-(n-1)}}{c_{n-1}^{\pr}}),
\end{equation}
where $\sum_{j=0}^{n-1} c_j^{\pr}\al^{-ij}=0$ for each $1 \le i \le \de-1$.

By \cite[Corollary 5.3.3]{AK}, $\RM^+(m-r,m)^*$ contains two minimum weight codewords $c_1(x)$ and $c_2(x)$, such that the support of $c_1(x)$ and $c_2(x)$ are $H_1 \setminus \{0\}$ and $H_2 \setminus  \{0\}$ respectively. Similarly, $\RM^-(m-r,m)^*$ contains two minimum weight codewords $c_3(x)$ and $c_4(x)$, such that the support of $c_3(x)$ and $c_4(x)$ are $H_3 \setminus  \{0\}$ and $H_4 \setminus  \{0\}$ respectively. Moreover, the coordinates of $c_1(x)$ and $c_2(x)$ are arranged in the way of~(\ref{eqn-1}) and the coordinates of $c_3(x)$ and $c_4(x)$ are arranged in the way of~(\ref{eqn-2}). Therefore,
$$
c_1(x)+c_2(x) \in \RM^+(m-r,m)^* \subset \C^+_{(2,m,\delta)}
$$
and
$$
c_3(x)+c_4(x) \in \RM^-(m-r,m)^* \subset \C^-_{(2,m,\delta)}.
$$
Since $((H_1 \cup H_2) \setminus \{0\})^{(-1)}=(H_3 \cup H_4) \setminus \{0\}$, by the arrangement of the coordinates of $c_i(x)$, $1 \le i \le 4$, the two codewords $c_1(x)+c_2(x)$ and $c_3(x)+c_4(x)$ coincide. Thus, we have $c_1(x)+c_2(x) \in \tilde{\C}_{(2,m,\delta)}$. Since $c_1(1)+c_2(1)=0$, we have a codeword $c_1(x)+c_2(x) \in \overline{\C}_{(2,m,\delta)}$ with weight $2\de$.
\end{proof}

\begin{example}
Let $q=2$, $m=5$ and $\de=3$. We consider the minimum distance of $\ol{\C}_{(2,5,3)}$. Let $\al$ be a primitive element of $\gf(2^5)$ and the minimal polynomial of $\al$ over $\Ft$ is $x^5+x^2+1$. Then we have the following four $2$-dimensional subspaces of $\gf(2^5)$:
\begin{align*}
H_1&=\{0,\al,\al^2,\al^{19}\}, \\
H_2&=\{0,\al^8,\al^{12},\al^{18}\}, \\
H_3&=\{0,\al^{12},\al^{13},\al^{30}\}, \\
H_4&=\{0,\al^{19},\al^{23},\al^{29}\}.
\end{align*}
Thus, we have $c_1(x)$ and $c_2(x)$ as codewords of $\C^+_{(2,5,3)}$, whose supports are $H_1 \setminus \{0\}$ and $H_2 \setminus \{0\}$. We have $c_3(x)$ and $c_4(x)$ as codewords of $\C^-_{(2,5,3)}$, whose supports are $H_3 \setminus \{0\}$ and $H_4 \setminus \{0\}$. Clearly, $((H_1 \cup H_2) \setminus \{0\})^{(-1)}=(H_3 \cup H_4) \setminus \{0\}$. Therefore, $c_1(x)+c_2(x)$ coincides with $c_3(x)+c_4(x)$, whose weight is six. Consequently, $c_1(x)+c_2(x) \in \ol{\C}_{(2,5,3)}$ and the minimum distance of $\ol{\C}_{(2,5,3)}$ equals $6$.
\end{example}

Based on our numerical experiment, we have the following conjecture, which can be regarded as an analogy of \cite[Chapter 9, Theorem 5]{MS77}.

\begin{conj}
Let $\delta=q^\lambda-1$, where $1 \leq \lambda \leq \lfloor m/2 \rfloor$. Then the code  $\ocode$ has minimum
distance $d=2\delta$.
\end{conj}

\section{The parameters of the reversible code $\ocode$ for small $\delta$}

In this section, we determine the parameters of the code $\ocode$ for a few small values of $\delta$. With the help of Theorem \ref{thm-gene} and Corollary \ref{cor-distance}, we can achieve this in some cases.

Recall that the Melas code over $\gf(q)$ is a cyclic code with length $n=q^m-1$ and generator polynomial $m_{-1}(x)m_1(x)$ and was first studied by Melas for the case $q=2$ \cite{Melas}. The weight distribution of the Melas code has been obtained for $q=2,3$ \cite{GSV,SV}. For $\delta = 2$, the code $\ocode$ is the even-like subcode of the Melas code. The following theorem is a direct consequence of Theorem \ref{thm-gene} and Corollary \ref{cor-distance}.

\begin{theorem}
Suppose $q$ is odd and $ m \ge 2$, then $\ol{\C}_{(q,m,2)}$ has parameters $[q^m-1, q^m-2-2m, 4]$.
\end{theorem}

When $\delta = 3$, we have the following result.

\begin{theorem}\label{thm-dl2}
\begin{itemize}
\item[1)] When $q=2$ and $m \geq 4$, $\ol{\C}_{(q,m,3)}$ has parameters $[2^m-1, 2^m-2-2m, 6]$.
\item[2)] When $q^m \equiv 1 \pmod{3}$ and $m \geq 4$,  $\ol{\C}_{(q,m,3)}$ has parameters $[q^m-1, q^m-2-4m, 6]$.
\end{itemize}
\end{theorem}
\begin{proof}
1) The dimension follows from Theorem \ref{thm-gene}. Applying the BCH and the sphere packing bound, we can see that the minimum distance is $6$.

2) The dimension follows from Theorem \ref{thm-gene}. Since $q^m\equiv 1 \pmod{3}$, we have $3 \mid n$. Therefore, by Corollary \ref{cor-distance}, the minimum distance is $6$.
\end{proof}

\begin{theorem}
Suppose $m \ge 3$, then $\ol{\C}_{(3,m,4)}$ has parameters $[3^m-1, 3^m-2-4m, d]$, where $d=8$ if
$m$ is even and $d \geq 8$ if $m$ is odd.
\end{theorem}

\begin{proof}
It follows from Theorem \ref{thm-gene} that the dimension of this code is equal to $q^m-2-4m$. By the BCH
bound, the minimum distance of $\overline{\C}_{(3,m,4)}$ is at least $8$.

When $m$ is even, $4$ divides $n$. Hence, the minimum distance of $\overline{\C}_{(3,m,4)}$ is equal to $8$ according to Corollary \ref{cor-distance}.
\end{proof}

We have the following conjecture concerning the case $q=3$.

\begin{conj}
When $\delta=4$, $m \geq 3$ is odd, and $q=3$,  $\ocode$ has minimum distance $d=8$.
\end{conj}

\begin{example}
Let $m=3$ and $q=3$. Let $\alpha$ be a generator of $\gf(3^3)^*$ with $\alpha^3-\alpha+1=0$.
Then $\ocode$ has parameters $[26,13, 8]$ and generator polynomial
$$
x^{13} + x^{12} + 2x^{11} + 2x^{10} + x^8 + 2x^5 + x^3 + x^2 + 2x + 2.
$$
This is the best ternary cyclic code and has the same parameters as the best linear code with the
same length and dimension.
\end{example}

\section{Concluding remarks}

In this paper, we determined the dimension of the reversible BCH code $\overline{\C}_{(q,\, m,\, \delta)}$ when $\delta$ is relatively small, and developed lower and upper bounds on its dimension for $\delta=q^\lambda$, where $\frac{m}{2} \le \lambda \le m-1$. We determined the minimum distance of $\ocode$ in two special cases. Moreover, we settled the parameters of the code $\overline{\C}_{(q,\, m,\, \delta)}$ for some very small $\delta$.

However, the parameters of the reversible BCH codes $\overline{\C}_{(q,\, m,\, \delta)}$ are still open in most cases. It would be nice if further progress on the determination of the parameters of the codes $\overline{\C}_{(q,\, m,\, \delta)}$ can be made. According to the tables of best cyclic codes documented in \cite{Dingbk15},   the binary reversible BCH code $\overline{\C}_{(q,\, m,\, \delta)}$ is the best possible cyclic code in almost all cases (see Table \ref{tab-revcodeBCH1}). Hence, it is worthwhile to study the code $\overline{\C}_{(q,\, m,\, \delta)}$.

\begin{table}[ht]
\begin{center}
\caption{Examples of binary reversible BCH code $\overline{\C}_{(2,\, m,\, \delta)}$}\label{tab-revcodeBCH1}
\begin{tabular}{|r|r|r|r|r|c|c|} \hline
$m$ & $n$ &   $k$   & $d$  & $\delta$ & Best cyclic?  & Optimal? \\ \hline
4 & 15 & 6 & 6 &  3  & Yes & Yes \\
4 & 15 & 2 & 10 & 5     & Yes & Yes \\
5 & 31 & 20 & 6 & 3     & Yes & Yes \\
5 & 31 & 10 & 10 & 5   & No & No \\
6 & 63 & 50 & 6 &  3     & Yes & Yes \\
6 & 63 & 38 & 10 & 5   & Yes & Unknown \\
6 & 63 & 26 & 14 & 7   & Yes & No \\
6 & 63 & 20 & 18 & 9   & Yes & Unknown \\
6 & 63 & 14 & 22 & 11  & Yes & No \\
6 & 63 & 2 & 42 & 13     & Yes & Yes \\ \hline
\end{tabular}
\end{center}
\end{table}

\end{document}